\documentclass[12pt]{article}
\usepackage[margin=1in]{geometry}
\usepackage[latin2]{inputenc}
\usepackage[english]{babel}
\usepackage{xspace}
\usepackage{amsmath,amssymb,amsthm,amsopn}
\usepackage{algorithm}
\usepackage{algorithmic}
\usepackage[ruled,vlined,algo2e,linesnumbered]{algorithm2e}
\usepackage{graphicx}
\usepackage{todonotes}
\usepackage{color}
\usepackage{authblk}

\usepackage{lipsum}
\usepackage{amsfonts}
\usepackage{graphicx}
\usepackage{epstopdf}
\usepackage{hyperref}

\newtheorem{theorem}{\bf Theorem}[section]
\newtheorem{lemma}{\bf Lemma}[section]

\newtheorem{proposition}{Proposition}

\newcommand{\OO}{{\mathcal O}}

\newcommand{\ETH}{\textsf{ETH} }
\newcommand{\Yes}{{\sc Yes}}
\newcommand{\yes}{{\sc Yes}}
\newcommand{\No}{{\sc No}}
\newcommand{\FPT}{\textrm{\textup{FPT}}}

\newcommand{\NP}{\textrm{\textup{NP}}}
\newcommand{\NPH}{\textrm{\textup{NP-hard}}}

\newcommand{\II}{\mathcal{I}}

\newcommand{\cO}{{\cal O}}

\newcommand{\mrank}[1]{\mbox{\sf rank}(#1)}

\newcommand{\matparity}{\textsc{$\alpha$-Matroid Parity}\xspace}
\newcommand{\maxsimfesfull}{\textsc{Maximum Simultaneous Acyclic Subgraph}\xspace}
\newcommand{\maxsimfes}{\textsc{Max-Sim-Subgraph}\xspace}

\newcommand{\minsimfes}{\textsc{Sim-FES}\xspace}
\newcommand{\simfesfull}{\textsc{Simultaneous Feedback Edge Set}}
\newcommand{\simfes}{\textsc{Sim-FES}}

\newcommand{\colfn}{{${\sf col}: E(G) \rightarrow 2^{[\alpha]}$}}
\newcommand{\coln}{{{\sf col}}}
\newcommand{\mat}{$M=(E,{\cal I})$}
\newcommand{\matl}[1]{$M_{#1}=(E_{#1},{\cal I}_{#1})$}

 \newcommand{\defparproblem}[4]{
  \vspace{3mm}
 \noindent\fbox{
   \begin{minipage}{.95\textwidth}
   \begin{tabular*}{\textwidth}{@{\extracolsep{\fill}}lr} \textsc{#1}  \\ \end{tabular*}
   {\bf{Input:}} #2  \\
    {\bf{Parameter:}} #3 \\
   {\bf{Question:}} #4
   \end{minipage}
   }
   \vspace{2mm}
 }

\begin{document}

\title{Simultaneous Feedback Edge Set: A Parameterized Perspective\thanks{A preliminary version of this paper will appear in the proceedings of the 27th International Symposium
Algorithms and Computation (ISAAC 2016).
The research leading to these results has received funding from the European Research Council (ERC)
via grants Rigorous Theory of Preprocessing, reference 267959 and PARAPPROX, reference 306992.
}}


\author[1]{ Akanksha Agrawal}
\author[1]{Fahad Panolan}
\author[1,2]{Saket Saurabh}
\author[1]{Meirav Zehavi}
\affil[1]{Department of Informatics, University of Bergen, Norway. \texttt{\{akanksha.agrawal|fahad.panolan|meirav.zehavi\}@ii.uib.no}}
\affil[2]{The Institute of Mathematical Sciences, HBNI, Chennai, India. \texttt{saket@imsc.res.in}}

\date{}


\maketitle

\begin{abstract}
In a recent article Agrawal et al.~(STACS 2016) studied a simultaneous variant of the classic {\sc Feedback Vertex Set} problem, called \textsc{Simultaneous Feedback Vertex Set (Sim-FVS)}. In this problem the input is an $n$-vertex graph $G$, an integer $k$ and a coloring function ${\sf col}: E(G) \rightarrow 2^{[\alpha]}$, and the objective is to check whether there exists  a vertex subset $S$ of cardinality at most $k$ in $G$ such that for all $i\in [\alpha]$, $G_i - S$ is acyclic. Here, $G_i=(V(G), \{e\in E(G) \mid i \in {\sf col}(e)\})$ and $[\alpha]=\{1,\ldots,\alpha\}$. In this paper we consider the edge variant of the problem, namely,  \simfesfull\  (\simfes). In this problem, the input is same as the input of \textsc{Sim-FVS} and the objective is to check whether there is an edge subset $S$ of cardinality at most $k$ in $G$ such that for all $i \in [\alpha]$, $G_i - S$ is acyclic. Unlike the vertex variant of the problem, when $\alpha =1$, the problem is 
equivalent to finding a maximal spanning forest and hence it is polynomial time solvable. We show that for $\alpha =3$ \simfes{} is \NP-hard by giving a reduction from \textsc{Vertex Cover} on cubic-graphs. The same reduction shows that the problem does not admit an algorithm  of running time $\OO(2^{o(k)}n^{\OO(1)})$  unless \ETH fails. This hardness result  is complimented by an \FPT{} algorithm 
for \simfes{} running in time $\OO(2^{\omega k \alpha+\alpha \log k} n^{\OO(1)})$, where 
$\omega$ is the exponent in the running time of matrix multiplication. The same algorithm gives a polynomial time algorithm for the case when $\alpha =2$. 
We also give a kernel for \simfes{} with $(k\alpha)^{\mathcal{O}(\alpha)}$ vertices. Finally, we consider the problem \maxsimfesfull. Here, the input is a graph $G$, an integer $q$ and, a coloring function ${\sf col}: E(G) \rightarrow 2^{[\alpha]}$. The question is whether there is a edge subset $F$ of cardinality at least $q$ in $G$ such that for all $i\in [\alpha]$, $G[F_i]$ is acyclic. Here, $F_i=\{e \in F \mid i \in \textsf{col}(e)\}$. We give an \FPT{} algorithm for \maxsimfesfull{} running in time $\OO(2^{\omega q \alpha}n^{\mathcal{O}(1)})$. All our algorithms are based on parameterized version of the  {\sc Matroid Parity} problem. 
\end{abstract}

%


\section{Introduction}\label{sec:intro}
Deleting  at most $k$ vertices or edges from  a given graph $G$, so that the resulting graph belongs to a particular family of graphs ($\cal F$), is an important research direction in the fields of 
graph algorithms and parameterized complexity.  For a family of graphs $\mathcal{F}$, given a graph $G$ and an integer $k$, the \textsc{$\mathcal{F}$-deletion} (\textsc{Edge $\mathcal{F}$-deletion}) problem asks whether we can delete at most $k$ vertices (edges) in $G$ so that the resulting graph belongs to $\mathcal{F}$. The \textsc{$\mathcal{F}$-deletion} (\textsc{Edge $\mathcal{F}$-deletion}) problems generalize many of the NP-hard problems like \textsc{Vertex Cover}, \textsc{Feedback vertex set}, \textsc{Odd cycle transversal}, {\sc Edge Bipartization}, etc.
Inspired by  applications, Cai and Ye introduced  variants of 
\textsc{$\mathcal{F}$-deletion} (\textsc{Edge $\mathcal{F}$-deletion}) problems on edge colored graph~\cite{caiye2014}. Edge colored graphs are studied in graph theory with respect to various problems like \textsc{Monochromatic and Heterochromatic Subgraphs}~\cite{kano-monochromatic}, \textsc{Alternating paths}~\cite{jensen-alternating,chou-paths,manoussakis-alternating}, Homomorphism in edge-colored graphs~\cite{alon-homorphism}, \textsc{Graph Partitioning} in 2-edge colored graphs~\cite{balogh-partitioning} etc. One of the natural generalization to the classic \textsc{$\mathcal{F}$-deletion}  (\textsc{Edge $\mathcal{F}$-deletion})  problems on edge colored graphs is the following. Given a graph $G$ with a coloring function ${\sf col}: E(G) \rightarrow 2^{[\alpha]}$, and an integer $k$, we want to delete a set $S$ of at most $k$ edges/vertices in $G$ so that for each $i \in [\alpha]$, $G_i - S$ belongs to $\mathcal{F}$. Here, $G_i$ is the graph with vertex set $V(G)$ and edge set as $\{e\in E(G) \mid i \in {\sf col}(e)\}$. These problems are also called {\em simultaneous} variant of \textsc{$\mathcal{F}$-deletion} (\textsc{Edge $\mathcal{F}$-deletion}). 

Cai and Ye studied the \textsc{Dually Connected Induced subgraph} and \textsc{Dual Separator} on 2-edge colored graphs~\cite{caiye2014}. Agrawal et al.~\cite{Sim-FVS} studied a 
simultaneous variant of {\sc Feedback Vertex Set} problem, called \textsc{Simultaneous Feedback Vertex Set}, in the realm of parameterized complexity.  Here, the input is a graph $G$, an integer $k$, and a coloring function ${\sf col}: E(G) \rightarrow 2^{[\alpha]}$ and the objective is to check whether there is  a set $S$ of at most $k$ vertices in $G$ such that for all $i\in [\alpha]$, $G_i - S$ is acyclic. Here, $G_i=(V(G), \{e\in E(G) \mid i \in {\sf col}(e)\})$. In this paper we consider the edge variant of the problem, namely,  \simfesfull{}, in the realm of parameterized complexity.

 In the Parameterized Complexity paradigm the main objective is to design an algorithm with running   time $f(\mu)\cdot n^{\OO(1)}$, where $\mu$ is the parameter associated with the input, $n$ is the size of the input and $f(\cdot)$ is some computable function whose value depends only on $\mu$. A problem which admits such an algorithm is said to be \emph{fixed parameter tractable} parameterized by $\mu$. Typically, for edge/vertex deletion problems one of the natural parameter that is associated with the input is the size of the solution we are looking for. Another objective in parameterized complexity is to design polynomial time pre-processing routines that reduces the size of the input as much as possible. The notion of such a pre-processing routine is captured  by \emph{kernelization} algorithms. The kernelization algorithm for a parameterized problem $Q$ takes as input an instance $(I,k)$ of $Q$, runs in polynomial time and returns an equivalent instance $(I',k')$ of $Q$. Moreover, the size of the instance $(I',k')$ returned by the kernelization algorithm is bounded by $g(k)$, where $g(\cdot)$ is some computable function whose value depends only on $k$. If $g(\cdot)$ is polynomial in $k$, then the problem $Q$ is said to admit a polynomial kernel. The instance returned by the kernelization is referred to as a \emph{kernel} or a reduced instance.  We refer the readers to the recent book of Cygan et al.~\cite{saket-book} for a more detailed overview of parameterized complexity and kernelization. 

A \emph{feedback edge set} in a graph $G$ is $S\subseteq E(G)$ such that $G- S$ is a forest. For a graph $G$ with a coloring function \colfn,  \emph{simultaneous feedback edge set} is a subset $S \subseteq E(G)$ such that $G_i -S $ is a forest for all $i \in [\alpha]$. Here, $G_i=(V(G), E_i)$, where $E_i=\{e \in E(G) \mid i \in {\sf col}(e)\}$. Formally, the problem is stated below.

\defparproblem{\simfesfull\  (\simfes)  }{An $n$-vertex graph $G$, $k\in {\mathbb N}$ and a coloring function \colfn}{$k,\alpha$} 
{Is there a simultaneous feedback edge set of cardinality at most $k$ in $G$} 

{\sc Feedback Vertex Set (FVS)} is one of the classic \NP-complete~\cite{GJ79} problems and has been extensively studied from all the algorithmic paradigms  that are meant for coping with \NP-hardness, such as approximation algorithms, parameterized complexity and moderately exponential time algorithms.  
The problem
admits a factor $2$-approximation algorithm~\cite{2-approx-fvs-bafna}, an exact 
algorithm with running time $\OO(1.7217^n  n^{\OO(1)})$~\cite{exactmonotoneFomin}, a deterministic parameterized algorithm of running in time $\OO(3.619^k n^{\OO(1)})$~\cite{Kociumaka2014556}, a randomized algorithm running
in $\OO(3^k n^{\OO(1)} )$ time~\cite{Cygan:2011:SCP:2082752.2082943},
and a kernel with $\OO(k^2)$ vertices~\cite{Thomasse:2010:KKF:1721837.1721848}.  Agrawal et al.~\cite{Sim-FVS} studied  \textsc{Simultaneous Feedback Vertex Set (Sim-FVS)} and gave an \FPT{} algorithm running in time $2^{\OO(\alpha k)}n^{\OO(1)}$ and a kernel of size $\OO(\alpha k^{3(\alpha +1)})$.
Finally, unlike the FVS problem, \simfes{} is polynomial time solvable when $\alpha=1$, because
it  is equivalent to finding maximal spanning forest.

\subsection*{Our results and approach}
In Section~\ref{sec:sim-fes-fpt} we design an \FPT{} algorithm for \minsimfes{} by 
reducing to  \matparity{} on direct sum of elongated co-graphic matroids of $G_i$, $i\in [\alpha]$ 
(see Section~\ref{sec:prelim} for definitions related to matroids). This algorithm runs in time $\OO(2^{\omega k \alpha+\alpha \log k} n^{\OO(1)})$. 
Unlike the vertex counterpart, we show that for $\alpha=2$ (2-edge colored graphs) \simfes{} is polynomial time solvable. This follows from the polynomial time algorithm for the \textsc{Matroid parity} problem. In Section~\ref{sec:hardness} we show that for $\alpha = 3$, \simfes{} is \NP-hard. Towards this, we give a reduction from the \textsc{Vertex Cover} in cubic graphs which is known to be \NP-hard~\cite{mohar2001face}. Furthermore, the same reduction shows that the problem cannot be solved in 
$2^{o(k)} n^{\OO(1)}$ time unless Exponential Time Hypothesis (\textsc{ETH}) fails~\cite{Impagliazzo:2001:PSE:569473.569474}. We complement our \FPT{} algorithms by showing that \simfes{} is W[1]-hard when parameterized by the solution size $k$ (Section~\ref{sec:simfes-lower-bounds}). When $\alpha= \OO(|V(G)|)$, we give a parameter preserving reduction from the \textsc{Hitting Set} problem, a well known W[2]-hard problem parameterized by the solution size~\cite{saket-book}. However, \simfes{} remains W[1]-hard even when $\alpha= \OO(\log (|V(G)|))$. We show this by giving a parameter preserving reduction from \textsc{Partitioned Hitting Set} problem, a variant of the \textsc{Hitting set} 
problem,  defined in~\cite{Sim-FVS}. In \cite{Sim-FVS},   \textsc{Partitioned Hitting Set} was shown to be W[1]-hard parameterized by the solution size. In Section~\ref{sec:kernel} we give a kernel with $\mathcal{O}((k\alpha)^{\mathcal{O}(\alpha)})$ vertices. Towards this we apply some of the standard preprocessing rules for obtaining kernel for  \textsc{Feedback Vertex Set} and use the approach similar to the one developed  for designing kernelization algorithm for \textsc{Sim-FVS}~\cite{Sim-FVS}. In Section~\ref{sec:exactalgo-fes} we give an \FPT\ algorithm for 
the problem, when parameterized by the dual parameter. Formally, this problem is defined as follows.

\defparproblem{\maxsimfesfull\  (\maxsimfes)}
{An $n$-vertex graph $G$, $q\in {\mathbb N}$ and a function \colfn.} {$q$}
{Is there a subset $F \subseteq E(G)$ such that $|F|\geq q$ and for all $i \in [\alpha]$, $G[F \cap E(G_i)]$ is acyclic?}

For solving \maxsimfes we reduce it to an equivalent instance of the \matparity problem. As an immediate corollary we get an exact algorithm for \simfes\ running in time $\cO(2^{\omega n \alpha^2} n^{\cO(1)})$. 

\section{Preliminaries}\label{sec:prelim}
We denote the set of natural numbers by $\mathbb{N}$. For $n\in \mathbb{N}$, by $[n]$ we denote the set $\{1,\ldots,n\}$. For a set $X$, by $2^X$ we denote the set 
of all subsets of $X$. We use the term \emph{ground set}/ \emph{universe} to distinguish a set from its subsets. We will use $\omega$ to denote the exponent in the running time of
matrix multiplication, the current best known bound for $\omega$ is $<2.373$~\cite{Williams12}. 

\subsection{Graphs}
We use the term {\em graph} to denote undirected graph. 
For a graph $G$, by $V(G)$ and $E(G)$ we denote its vertex set and edge set, respectively. We will be considering finite graphs possibly having loops and multi-edges.
In the following, let $G$ be a graph and let $H$ be a subgraph of $G$. By $d_{H}(v)$, we denote the
degree of the vertex $v$ in $H$, i.e, the number of edges in H which are incident with $v$. 
A self-loop at a vertex $v$ contributes $2$ to the degree of $v$.
For any non-empty subset $W \subseteq V(G)$, the subgraphs of $G$ induced by $W$, $V(G)\setminus W$ are denoted by $G[W]$ and $G-W$ respectively. Similarly, for $F \subseteq E(G)$, the subgraph of $G$ induced by $F$ is denoted by $G[F]$; its vertex set is $V(G)$ and its
edge set is $F$. For $F \subseteq E(G)$, by $G - F$ we denote the graph obtained by
deleting the edges in $F$. We use the convention that a double edge and a self-loop is a cycle.  An $\alpha$-edge colored graph is a graph $G$ with a color function \colfn{}. By $G_i$ we will denote the color $i$ (or $i$-color) graph of $G$, where $V(G_i)=V(G)$ and $E(G_i)=\{e\in E(G) | i \in {\sf col}(e)\}$. For an $\alpha$-edge colored graph $G$, the \emph{total degree} of a vertex
$v$ is $\sum_{i=1}^{\alpha} d_{G_i}(v)$. We refer the reader to~\cite{diestel-book} for details on standard graph theoretic notations and terminologies.

\subsection{Matroids}
A pair \mat, where $E$ is a ground set and $\cal I$ is a family of subsets (called independent sets) of $E$, is a {\em matroid}
if it satisfies the following conditions: 
\begin{itemize}
 \item[\rm (I1)]  $\phi \in \cal I$,
 \item[\rm (I2)]  if $A' \subseteq A $ and $A\in \cal I$ then $A' \in  \cal I$, and 
 \item[\rm (I3)] if $A, B  \in \cal I$  and $
|A| < |B| $, then there is $ e \in  (B \setminus A) $ such that $A\cup\{e\} \in \cal I$. 
\end{itemize}
 
The axiom (I2) is also called the hereditary property and a pair $(E,\cal I)$  satisfying  only (I2) is called hereditary family. An inclusion wise maximal subset of $\cal I$ is
called a {\em basis} of the matroid. Using axiom (I3) it is easy to show that all the bases of a matroid have the same size. This size is called the {\em rank} of the matroid $M$,
and is denoted by \mrank{M}. We refer the reader to~\cite{oxleymatroid} for more details about matroids.

\subparagraph{Representable Matroids} 
Let $A$ be a matrix over an arbitrary field $\mathbb F$ and let $E$ be the set of columns of $A$. For $A$, we define  matroid 
 \mat{} as follows. A set $X \subseteq E$ is independent (that is $X\in \cal I$) if the corresponding columns are linearly independent over $\mathbb F$. The matroids that can be defined by such a construction are called {\em linear matroids}, and if a matroid can be defined by a matrix $A$ over a 
field $\mathbb F$, then we say that the matroid is representable over $\mathbb F$. A matroid \mat{}  is called {\em representable} or {\em linear} if it is representable over some field $\mathbb F$. 

\subparagraph*{Direct Sum of Matroids} 
Let \matl{1}, \matl{2}, \dots, \matl{t} be $t$ matroids with $E_i\cap  E_j =\emptyset$ for all $1\leq i\neq j \leq t$. The direct sum
$M_1\oplus \cdots \oplus M_t$  is a matroid \mat{} with  $E := \bigcup_{i=1}^t E_i$ and $X\subseteq E$ is independent if and only if  $X\cap E_i\in {\cal I}_i$ for all $i\in [t]$. 
Let $A_i$  be the representation matrix of \matl{i} over field $\mathbb{F}$. 
Then,  
    \begin{displaymath}
    A_M=\left(
     \begin{array}{ccccc}
        A_1 & 0 & 0 & \cdots & 0\\
        0 &  A_2 &  0 & \cdots & 0 \\
        \vdots &   \vdots  &  \vdots   &  \vdots  &  \vdots \\
        0 & 0 & 0 & \cdots & A_t 
     \end{array} \right)
     \end{displaymath}
 is a representation matrix of $M_1\oplus \cdots \oplus M_t$. The correctness of this  is proved in~\cite{marxmatroid06,oxleymatroid}.

\subparagraph*{Uniform Matroid} 
A pair \mat{} over an $n$-element ground set $E$, is called a uniform matroid if the family of independent sets  is given by ${\cal I}=\{A\subseteq E~|~|A|\leq k\}$, where $k$ is some constant. This matroid is also denoted as $U_{n,k}$. 

\begin{proposition}[\cite{saket-book,oxleymatroid}] 
\label{prop:uniformrep}
Uniform matroid $U_{n,k}$ is representable over any field of size strictly more than $n$ and such a representation can be found in time polynomial in $n$.
\end{proposition}

\subparagraph*{Graphic and Cographic Matroid} Given a graph $G$, the graphic matroid \mat{} is defined by taking the edge set $E(G)$ as universe and $F\subseteq E(G)$ is in $\cal
I$ if and only if $G[F]$ is a forest.  Let $G$ be a graph and $\eta$ be the number of components in $G$. The co-graphic matroid \mat{} of $G$ is defined by taking the the edge set
$E(G)$ as universe and $F \subseteq E(G)$ is in $\cal I$ if and only if the number of connected components in $G - F$ is $\eta$. 

\begin{proposition}[\cite{oxleymatroid}]
\label{prop:graphicrep}
Graphic and co-graphic matroids are representable over any field of size  $\geq 2$ and such a representation can be found in time polynomial in the size of the  graph.
\end{proposition}

\subparagraph*{Elongation of Matroid}
Let $M=(E, \II)$ be a matroid and $k$ be an integer such that $\mrank{M} \leq k \leq |E|$. A $k$-elongation matroid $M_k$ of $M$ is a matroid with the universe as $E$ and $S
\subseteq E$ is a basis of $M_k$ if and only if, it contains a basis of $M$ and $|S| = k $. Observe that the rank of the matroid $M_k$ is $k$.

\begin{proposition}[\cite{lokshtanov2015deterministic}]\label{prop:eolongation-representation}
Let $M$ be a linear matroid of rank $r$, over a ground set of size $n$, which is representable over a field $\mathbb{F}$. Given a number $\ell \geq r$, we can compute a representation of the $\ell$-elongation of $M$, over the field $\mathbb{F}(X)$ in $\mathcal{O}(nr\ell)$ field operations over $\mathbb{F}$.
\end{proposition}

\subparagraph*{$\alpha$-Matroid Parity}
In our algorithms we use a known algorithm for \matparity. Below we define \matparity{} problem formally and 
state its algorithmic result. 

\defparproblem{\matparity}
{A representation $A_M$ of a linear matroid \mat{}, a partition $\mathcal{P}$ of $E$ into blocks of size $\alpha$ and a positive integer $q$.} {$\alpha,q$}
{Does there exist an independent set which is a union of $q$ blocks?}

\begin{proposition}[\cite{lokshtanov2015deterministic,marxmatroid06}]
\label{alg:matparity}
There is a deterministic algorithm for \matparity, running in time $\cO(2^{\omega q\alpha} ||A_M||^{\cO(1)})$, 
where $||A_M||$ is the total number of bits required to describe all the elements of matrix $A_M$.
\end{proposition}

\section{FPT Algorithm for \simfesfull}
\label{sec:sim-fes-fpt}
In this section we design an algorithm for \minsimfes{} 
by giving a reduction to \matparity\ on the direct sum of elongated co-graphic matroids 
associated with graphs restricted to different color classes.   

We describe our algorithm, {\sf Algo-SimFES}, for \minsimfes. Let $(G,k,{\sf col} : E(G)\rightarrow 2^{[\alpha]})$ be an input instance to \minsimfes. Recall that for $i\in [\alpha]$, $G_i$ is the graph with vertex set as $V(G)$ and edge set as $E(G_i)=\{e\in E(G) \mid i\in {\sf col}(e)\}$. Let $n=\vert V(G)\vert$. Note that $n=\vert V(G_i)\vert$ for all $i\in [\alpha]$. 
Let $\eta_i$ be the number of connected components in $G_i$. 
To make $G_i$ acyclic we need to delete at least $\vert E(G_i)\vert - n + \eta_i$  edges from $G_i$. Therefore, if there is $i \in [\alpha]$ such that $\vert E(G_i)\vert -n + \eta_i > k $, then {\sf Algo-SimFES} returns \No. We let $k_i = \vert E(G_i)\vert - n + \eta_i$. 
Observe that for $i \in [\alpha]$, $0 \leq k_i \leq k$. We need to delete at least $k_i$ edges from $E(G_i)$ to make $G_i$ acyclic. 
Therefore, the algorithm {\sf Alg-SimFES} for each $i \in [\alpha]$, guesses $k'_i$, where $k_i \leq k'_i \leq k$ and computes a solution $S$ of \minsimfes{} such that $\vert S \cap E(G_i)  \vert=k_i'$. Let $M_i=(E_i,{\cal I}_i)$ be the $k_i'$-elongation of the co-graphic matroid associated with $G_i$. 


\begin{proposition}
\label{prop:basisofelengatedcographic}
Let $G$ be a graph with $\eta$ connected components and $M$ be an $r$-elongation of 
the co-graphic matroid associated with $G$, where $r \geq \vert E(G)\vert -|V(G)|+\eta$. 
Then $B \subseteq E(G)$  is a basis of $M$ if and only if the subgraph $G-B$ is acyclic 
and $\vert B \vert=r$.  
\end{proposition}

\begin{proof}
In the forward direction let $B \subseteq E(G)$ be a basis of $M$.  By Definition of $M$ it follows that $|B|=r$ and $B$ contains a basis $B_c$ of the co-graphic matroid of $G$.
Suppose $G - B$ has a cycle. This implies that $G - B_c$ has a cycle. But then, there is an edge $e \in E(G-B_c)$ whose removal from $G - B_c$ does not increase the number of
connected components in $G -B_c$. This contradicts that $B_c$ was a basis in the co-graphic matroid of $G$.

In the reverse direction let $B\subseteq E(G)$ such that $|B|=r$ and $G - B$ is acyclic. Consider a inclusion wise maximal subset $B' \subseteq B$ such that the number of connected
components in $G -B'$ is $\eta$. Observe that $G - B'$ does not contain a cycle since $G - B$ is acyclic and $B'$ is inclusion wise maximal. Therefore, it follows that $B'$ is a
basis in the co-graphic matroid of $G$. But then $B$ contains a basis of the co-graphic matroid of $G$ and $|B|=r$, therefore $B$ is a basis in $M$.
\end{proof}

By Proposition~\ref{prop:basisofelengatedcographic}, for any basis $F_i$ in $M_i$, $G_i-F_i$ is acyclic. Therefore, our objective is to compute $F \subseteq E(G)$ such that $|F| = k$ and the elements of $F$ restricted to the elements of $M_i$ form a basis for all $i\in [\alpha]$. For this we will construct an instance of \matparity\ as follows. For each $e\in E(G)$ and $i \in \coln(e)$, we use $e^i$ to denote the corresponding element in $M_i$. 
For each  $e\in E(G)$, by  ${\sf Original}(e)$ we denote the set of elements $\{ e^j \mid  j\in \coln(e) \}$.  
For each edge $e \in E(G)$,  we define ${\sf Fake}(e)=\{ e^j \mid  j\in [\alpha]-\coln(e)\}$. Finally, for each edge $e\in E(G)$, by ${\sf Copies}(e)$ we denote the set ${\sf Original}(e) \cup {\sf Fake}(e)$. Let ${\sf Fake}(G)=\bigcup_{e\in E(G)}{\sf Fake}(e)$. Furthermore, let $\tau=|{\sf Fake}(G)|=\sum_{e\in E(G)} |{\sf Fake}(e)|$ and $k'= \sum_{i \in [\alpha]}(k-k'_i)$. Let $M_{\alpha+1}=(E_{\alpha+1},{\cal I}_{\alpha+1})$ be a uniform matroid over the ground set ${\sf Fake}(G)$. That is, $M_{\alpha+1}=U_{\tau,k'}$. By Propositions~\ref{prop:uniformrep} to 
 Proposition~\ref{prop:eolongation-representation} we know that $M_i$s are representable over $\mathbb{F}_{p}(X)$, where $p >\max(\tau,2)$ is a prime number  and their representation can be computed in polynomial time.  
Let $A_{i}$ be the linear representation of $M_{i}$ for all $i\in[\alpha+1]$.
 Notice that $E_i\cap  E_j =\emptyset$ for all $1\leq i\neq j \leq \alpha+1$. Let $M$ denote the direct sum 
$M_1\oplus \cdots \oplus M_{\alpha+1}$ with its representation matrix being $A_M$. Note that the ground set of $M$ is $\bigcup_{e\in E(G)} {\sf Copies}(e)$. Now we  define an instance of \matparity, which is the linear representation  $A_M$ of $M$ and  
the partition of ground set into ${\sf Copies}(e)$, $e\in E(G)$. 
Notice that for all $e\in E(G)$, $\vert {\sf Copies}(e)\vert=\alpha$. 
Also for each $i\in [\alpha]$, $\mrank{M_i}=k_i'$ and $\mrank{M_{\alpha+1}}=k'=\sum_{i\in [\alpha]}(k-k_i')$. This implies that $\mrank{M}=\alpha k$.

 Now {\sf Algo-SimFES} outputs \Yes\ if  there is a basis (an independent set of cardinality $\alpha k$) of $M$ which is a union of $k$ blocks in 
$M$  and otherwise outputs \No. {\sf Algo-SimFES} uses the algorithm mentioned in Proposition~\ref{alg:matparity} to check 
whether there is an independent set of $M$, composed of blocks.  
A pseudocode of {\sf Algo-SimFES} can be found in Algorithm~\ref{algoSimFES}. 

\begin{algorithm}[t]
\KwIn{ A graph $G,k\in {\mathbb N}$ and \colfn.}
\KwOut{\Yes{} if  there is a simultaneous feedback edge set of size $\leq k$ and 
\No{} otherwise.}
Let $\eta_i$ be the number of connected components in $G_i$ for all $i\in [\alpha]$\\
$k_i:= \vert E(G_i)\vert - n+\eta_i$ for all $i\in [\alpha]$\\
\If{there exists $i\in[\alpha]$ such that $k_i >k$}{\KwRet~\No \label{step:intialNO}}
\For{$(k_1',\ldots,k_{\alpha}') \in ([k]\cup \{0\})^{\alpha}$ such that $k_i\leq k_i'$ for all $i\in [\alpha]$}{
Let $M_i$ be the $k_i'$-elongation of the co-graphic matroid associated with $G_i$.\\ Let $M_{\alpha + 1}=U_{\tau,k'}$ over the gound set {\sf Fake}(G), where, $k'= \sum_{i \in [\alpha]} (k-k'_i)$. \\
Let  $M:=\bigoplus_{i\in [\alpha + 1]}M_i$.\\
For each $e\in E(G)$, let ${\sf Copies}(e)$ be the block of elements of $M$.\\
\If{there is an independent set of $M$ composed of $k$ blocks \label{step:basiscomputation} \label{step:blackbox}}{
\KwRet~\Yes\label{step:output1}\\
}
}
\KwRet~\No
  \caption{Pseudocode of {\sf Algo-SimFES}}
\label{algoSimFES}
\end{algorithm}

\begin{lemma} 
{\sf Algo-SimFES} is correct. 
\end{lemma}
\begin{proof} 
Let $(G,k,{\sf col} : E(G)\rightarrow 2^{[\alpha]})$ be a \Yes{} instance of 
\minsimfes and let $F\subseteq E(G)$, where $|F| = k$ be a solution of 
$(G,k,{\sf col} : E(G)\rightarrow 2^{[\alpha]})$.  
Let $k_i=\vert E(G_i)\vert-n+\eta_i$, where $\eta_i$ is the number of connected components in $G_i$, for all $i\in [\alpha]$. 
For all $i\in [\alpha]$, let 
$k_i'=\vert F \cap E(G_i)\vert$. 
Since $F$ is a solution, $k_i\leq k_i'$ 
 for all $i\in [\alpha]$. This implies that {\sf Algo-SimFES} will not execute Step~\ref{step:intialNO}. Consider the {\bf for} loop for the choice $(k_1',\ldots,k_{\alpha}')$. We claim that the columns corresponding to $S=\bigcup_{e\in F}{\sf Copies}(e)$ form a basis in $M$ and it is union of $k$ blocks. Note that $|S|= \alpha k$ by construction. For all $i\in [\alpha]$, let $F^i=\{e^i \mid e \in F,  i\in{\sf col}(e)\}$, which is subset of ground set of $M_i$.  
By Proposition~\ref{prop:basisofelengatedcographic}, for all $i\in[\alpha]$, $F^i$ is a basis for $M_i$. This takes care of all the edges in $\cup_{e\in F} {\sf Original}(e)$. Now let $S^*=S-\cup_{i\in [\alpha]} F^i =\cup_{e\in F} {\sf Fake}(e)$. Observe that $|S^*|= \sum_{i \in [\alpha]}(k-k'_i)=k'$. Also, $S^*$ is a subset of ground set of $U_{\tau, k'}$ and thus is a basis since $|S^*|= k'$. Hence $S$ is a basis of $M$. Note that $S$ is the union of blocks corresponding to $e \in F$ and hence is union of $k$ blocks. Therefore, {\sf Algo-SimFES} will output \Yes.

In the reverse direction suppose {\sf Algo-SimFES} outputs \Yes. This implies that there is a basis, say $S$, that is the union of $k$ blocks. By construction $S$ corresponds to union of the sets ${\sf Copies}(e)$ for some $k$ edges in $G$. Let these edges be $F=\{e_1,\ldots,e_k\}$. We claim that $F$ is  a solution of $(G,k,{\sf col} : E(G)\rightarrow 2^{[\alpha]})$. Clearly $\vert F \vert =k$. Since $S$ is a basis of $M$, for each $i \in [\alpha]$, $B(i) = S \cap \{e^i \mid e \in E(G_i)\}$ is a basis in $M_i$. Let $F(i)=\{e \mid e^i \in B(i)\}\subseteq F$. Since $B(i)$ is a basis of $M_i$, by Proposition~\ref{prop:basisofelengatedcographic}, $G_i-F(i)$ is an acyclic graph. 
\end{proof}

\begin{lemma}
{\sf Algo-SimFES} runs in deterministic time $\OO(2^{\omega k \alpha+\alpha \log k} \vert V(G)\vert^{\OO(1)})$.
\end{lemma}
\begin{proof}
The {\sf for} loop runs $(k+1)^{\alpha}$ times. The step~\ref{step:blackbox}  
uses the algorithm mentioned in 
 Proposition~\ref{alg:matparity}, which takes time 
$\OO(2^{\omega k \alpha} \vert \vert A_M\vert\vert^{\OO(1)})=\OO(2^{\omega k \alpha} \vert V(G)\vert^{\OO(1)})$. All other 
steps in the algorithm takes polynomial time. Thus, the total running time 
is $\OO(2^{\omega k \alpha+\alpha\log k} \vert V(G)\vert^{\OO(1)})$.
\end{proof}

Since \matparity\ for $\alpha=2$ can be solved in polynomial time~\cite{lovasz1980matroid} algorithm {\sf Algo-SimFES} 
runs in polynomial time for $\alpha=2$.
This gives us the following theorem. 
\begin{theorem}
\minsimfes{} is in \FPT{}  
and when $\alpha=2$ \minsimfes{} is in {\sf P}. 
\end{theorem}

\section{Hardness results for \simfes}
\label{sec:hardness}
In this section we show that when $\alpha=3$, \simfes{} is \NP-Hard. Furthermore, from our reduction we conclude that it is unlikely that {\sc Sim-FES} admits a subexponential-time algorithm. We give a reduction from {\sc Vertex Cover (VC)} in cubic graphs to the special case of {\sc Sim-FES} where $\alpha=3$. Let $(G,k)$ be an instance of VC in cubic graphs, which asks whether the graph $G$ has a  vertex cover of size at most $k$. We assume without loss of generality that $k \leq |V(G)|$. It is known that VC in cubic graphs is \NPH~\cite{mohar2001face} and unless the \ETH fails, it cannot be solved in time $\mathcal{O}^{\star}(2^{o(|V(G)|+|E(G)|)})$\footnote{$\mathcal{O}^{\star}$ notation suppresses polynomial factors in the running-time expression.}~\cite{komusiewicz2015tight}. Thus, to prove that when $\alpha=3$, it is unlikely that {\sc Sim-FES} admits a parameterized subexponential time algorithm (an algorithm of running time $\OO^{\star}(2^{o(k)})$), it is sufficient to construct (in polynomial time) an instance of the form $(G',k'=\mathcal{O}(|V(G)|+|E(G)|),{\sf col}': E'\rightarrow2^{[3]})$ of {\sc Sim-FES} that is equivalent to $(G,k)$. Refer Figure~\ref{fig:hardness1} for an illustration of the construction.

To construct $(G',k',{\sf col}': E(G') \rightarrow 2^{[3]})$, we first construct an instance $(\widehat{G},\widehat{k})$ of VC in subcubic graphs which is equivalent to $(G,k)$. 
We set 
$$V(\widehat{G})=V(G) \cup (\bigcup_{\{v,u\}\in E(G)}\{x_{v,u},x_{u,v}\})\mbox{, and}$$ 
$$E(\widehat{G})=\{\{x_{v,u},x_{u,v}\}: \{v,u\}\in E(G)\}\cup(\bigcup_{\{v,u\}\in
E(G)}\{\{v,x_{v,u}\},\{u,x_{u,v}\}\}).$$ 
That is, the graph $\widehat{G}$ is obtained from the graph $G$ by subdividing each edge in $E(G)$ twice. 


\begin{lemma}
\label{lem:vc-equivalence}
$G$ has a vertex cover of size $k$ if and only if $\widehat{G}$ has a vertex cover of size $\widehat{k}=k + |E(G)|$
\end{lemma}
\begin{proof}
In the forward direction, let $S$ be a vertex cover in $G$. We will construct a vertex cover $\widehat{S}$ in $\widehat{G}$ of size at most $k+|E(G)|$. Consider an edge $\{v,u\}
\in E(G)$. If both $u,v$ belongs to $S$, then we arbitrarily add one of the vertices from $\{x_{v,u},x_{u,v}\}$ to $\widehat{S}$. If exactly one of $\{v,u\}$ belongs to $S$, say $v
\in S$ then, we add $x_{u,v}$ to $\widehat{S}$. If $u\in S$, then we add $x_{v,u}$ to $\widehat{S}$. Clearly, $\widehat{S}$ is a vertex cover in $\widehat{G}$ and is of size at
most $k+ |E(G)|$.

In the reverse direction, given a vertex cover in $\widehat{G}$. For each $\{v,u\}\in E(G)$ such that both $x_{v,u}$ and $x_{u,v}$ are in the vertex cover, we can replace $x_{u,v}$
by $u$, and then, by removing all of the remaining vertices of the form $x_{v,u}$ (whose number is exactly $|E(G)|$), we obtain a vertex cover of $G$.
\end{proof}

Observe that in $\widehat{G}$ every path between two degree-3 vertices contains an edge of the form $\{x_{v,u},x_{u,v}\}$. Thus, the following procedure results in a partition $(M_1,M_2,M_3)$ of ${E}(\widehat{G})$ such that for all $i\in [3]$, $\{v,u\}\in M_i$ and $\{v',u'\}\in M_i\setminus \{\{v,u\}\}$, it holds that $\{v,u\}\cap\{v',u'\}=\emptyset$. Initially, $M_1=M_2=M_3=\emptyset$. For each degree-3 vertex $v$, let $\{v,x\}$, $\{v,y\}$ and $\{v,z\}$ be the edges containing $v$. We insert $\{v,x\}$ into $M_1$, $\{v,y\}$ into $M_2$, and $\{v,z\}$ into $M_3$ (the choice of which edge is inserted into which set is arbitrary). Finally, we insert each edge of the form $\{x_{v,u},x_{u,v}\}$ into a set $M_i$ that contains neither $\{v,x_{v,u}\}$ nor $\{u,x_{u,v}\}$.

We are now ready to construct the instance $(G',k',{\sf col}': E(G') \rightarrow 2^{[3]})$. Let $V(G')=V(\widehat{G})\cup V^{\star}$, where $V^{\star}=\{v^{\star}: v\in V(\widehat{G})\}$ contains a copy $v^{\star}$ of each vertex $v$ in $V(\widehat{G})$. The set $E(G')$ and coloring ${\sf col}'$ are constructed as follows. For each vertex $v\in V(\widehat{G})$, add an edge $\{v,v^{\star}\}$  into $E(G')$ and its 
color-set is $\{1,2,3\}$. For each $i\in [3]$ and for each $\{v,u\}\in M_i$, add the edges $\{v,u\}$ and $\{v^{\star},u^{\star} \}$  into $E(G')$ 
and its color-set is $\{i\}$. We set $k'=\widehat{k}$. Clearly, the instance $(G',k',{\sf col}': E(G') \rightarrow 2^{[3]})$ can be constructed in polynomial time, and it holds that $k'=\OO(|V(G)|+|E(G)|)$. 

Lemma~\ref{lem:vc-colored-graph} proves that $(\widehat{G},\widehat{k})$ is a \yes\ instance of VC if and only if $(G',k',{\sf col}': E(G') \rightarrow 2^{[3]})$ is a \yes\ instance of \simfes{}. Observe that because of the above mentioned property of the partition $(M_1,M_2,M_3)$ of $E(\widehat{G})$, we ensure that in $G'$, no vertex participates in two (or more) monochromatic cycles 
that have the same color. 
By construction, each monochromatic cycle in $G'$ is of the form $v-v^{\star}-u^{\star}-u-v$, where $\{v,u\} \in E(\widehat{G})$, and for each edge $\{v,u\}\in E(G')$, where either $v,u\in V(\widehat{G})$ or $v,u\in V^{\star}$, $G'$ contains exactly one monochromatic cycle of this form. 

\begin{lemma}\label{lem:vc-colored-graph} 
$(\widehat{G},\widehat{k})$ is a \yes\ instance of VC if and only if $(G',k',{\sf col}': E(G') \rightarrow 2^{[3]})$ is a \yes\ instance of \simfes.
\end{lemma}

\begin{proof}
In the forward direction, let $U$ be a vertex cover in $\widehat{G}$ of size at most $\widehat{k}$. Define $Q$ as the set of edges $\{\{v,v^{\star} \}: v\in U\}\subseteq E(G')$. We
claim that $Q$ is a solution to $(G',k',{\sf col}': E(G') \rightarrow 2^{[3]})$. Since $|Q|=|U|$, it holds that $|Q|\leq \widehat{k}=k'$. Now, consider a monochromatic cycle in
$G'$. Recall that such a cycle is of the form $v-v^{\star}-u^{\star}-u-v$, where $\{v,u\}\in E(\widehat{G})$. Since $U$ is a vertex cover of $\widehat{G}$, it holds that
$U\cap\{v,u\}\neq\emptyset$, which implies that $Q\cap\{\{v,v^{\star} \},\{u,u^{\star} \}\}\neq\emptyset$. 

In the reverse direction, let $Q$ be a solution to $(G',k',{\sf col}')$. 
Recall that for each edge $\{v,u\}\in E(G')$, where either $v,u\in V(\widehat{G})$ or $v,u\in V^{\star}$, $G'$ contains exactly one monochromatic cycle of this form. Therefore, if
$Q$ contains an edge of the form $\{v,u\}$ or of the form $\{v^{\star},u^{\star} \}$, such an edge can be replaced by the edge $\{v,v^{\star} \}$. Thus, we can assume that $Q$ only
contains edges of the form $\{v,v^{\star}\}$. Define $U$ as the set of vertices $\{v: \{v,v^{\star} \}\in Q\}\subseteq V(\widehat{G})$. We claim that $U$ is a vertex cover of
$\widehat{G}$ of size 
at most $\widehat{k}$. Since $|U|\leq |Q|$, it holds that $|U|\leq \widehat{k}$. Now, recall that for each edge $\{v,u\}\in E(\widehat{G})$, $G'$ contains a monochromatic cycle of
the form $v-v^{\star}-u^{\star}-u-v$. Since $Q$ is a solution to $(G',k',{\sf col}')$, it holds that $Q\cap\{\{v,v^{\star} \},\{u,u^{\star} \}\}\neq\emptyset$, which implies that 
$U\cap\{v,u\}\neq\emptyset$. 
\end{proof}

 \begin{figure}
 \medskip
 \centering
 \frame{
 \includegraphics[scale=0.5]{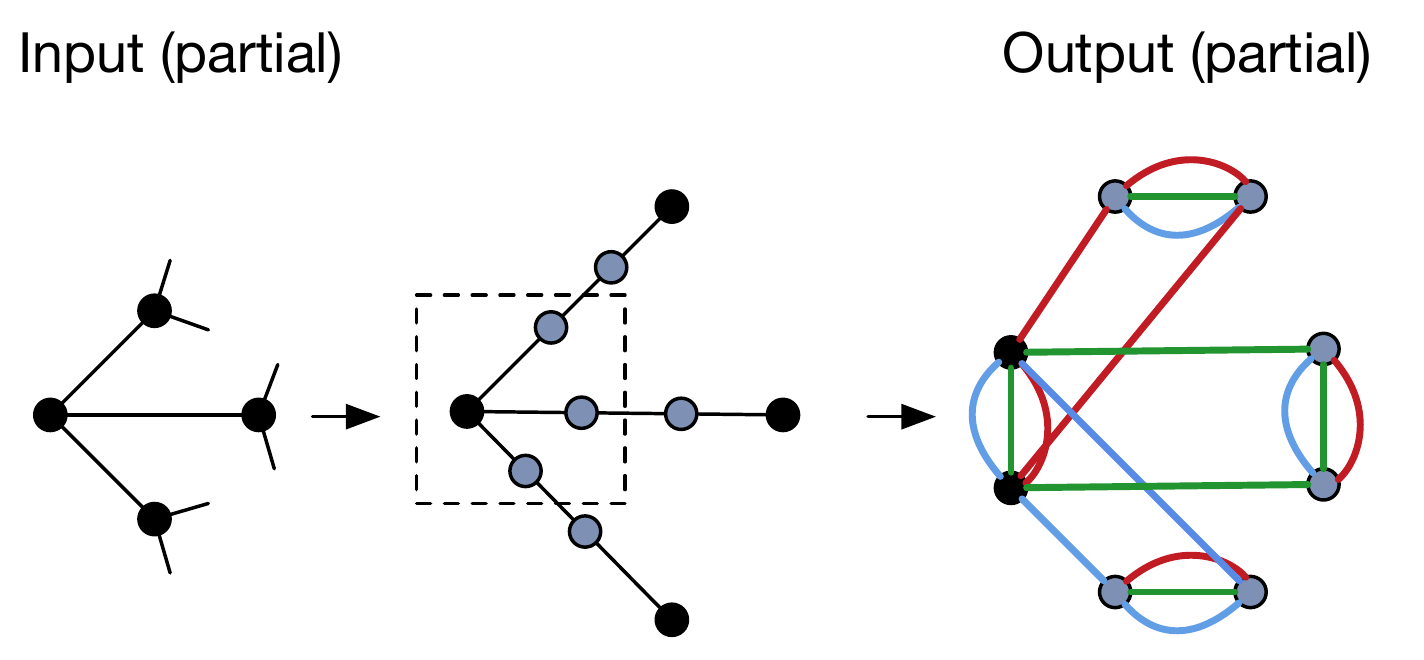}
 }
 \caption{The construction given in the proof of Theorem \ref{thm:NP_hard}.}
 \label{fig:hardness1}
 \end{figure}

We get the following theorem and its proof follows from  Lemma~\ref{lem:vc-equivalence} and Lemma~\ref{lem:vc-colored-graph}.

\begin{theorem}
\label{thm:NP_hard}
{\sc Sim-FES} where $\alpha=3$ is \NPH. Furthermore, unless the Exponential Time Hypothesis ({\sc ETH}) fails, 
{\sc Sim-FES} when $\alpha=3$ cannot be solved in time $\OO^*(2^{o(k)})$.
\end{theorem}

\section{Tight Lower Bounds for \simfes}\label{sec:simfes-lower-bounds}
We show that \simfes{} parameterized by $k$ is $W[2]$ hard when $\alpha= \mathcal{O}(|V(G)|)$ and $W[1]$ hard when $\alpha= \mathcal{O}(\log (|V(G)|))$. Our reductions follow the approach of Agrawal et al.~\cite{Sim-FVS}.

\subsection{W[2] Hardness of \simfes{} when $\alpha=\mathcal{O}(|V(G)|)$}
We give a reduction from {\sc Hitting Set (HS)} problem  to {\sc Sim-FES} where $\alpha=\OO(|V(G)|)$. Let 
$(U=\{u_1,\ldots,u_{|U|}\}, {\cal F}=\{F_1,\ldots,F_{|{\cal F}|}\},k)$ be an instance of HS, where ${\cal F}\subseteq 2^U$, which asks whether there exists a subset $S\subseteq U$ of size at most $k$ such that for all $F\in{\cal F}$, $S\cap F\neq\emptyset$. It is known that HS 
parameterized by $k$ is W[2]-hard (see, e.g., \cite{saket-book}). Thus, to prove the result, it is sufficient to construct (in polynomial time) an instance of the form $(G,k,{\sf col}: E(G) \rightarrow 2^{[\alpha]})$ of {\sc Sim-FES} that is equivalent to $(U,{\cal F},k)$, where $\alpha=\OO(|V(G)|)$. 
We construct a graph $G$ such that $V(G)=\OO(\vert U\vert \vert {\cal F} \vert)$ and the number of colors used will be 
$\alpha=\vert {\cal F}\vert$. The intuitive idea is to have one edge per element in the universe which is colored 
with all the indices of sets in the family ${\cal F}$ that  contains the element and for each $F_i\in {\cal F}$ creating a unique monochromatic cycle with color 
$i$ which passes through all the edges corresponding to the elements it contain. We explain the 
reduction formally in the next paragraph. 

Without loss of generality we assume that each set in ${\cal F}$ contains at least two elements from $U$. The instance $(G,k,{\sf col}: E(G) \rightarrow 2^{[\alpha]})$ is constructed as follows. Initially, $V(G)=E(G)=\emptyset$. For each element $u_i\in U$, insert two new vertices into $V(G)$, $v_i$ and $w_i$, add the edge $\{v_i,w_i\}$ into $E(G)$ and let $\{j \mid F_j\in {\cal F}, u_i\in F_j\}$ be its color-set. Now, for all $1\leq i<j\leq |U|$ and for all $1\leq t\leq |\mathcal{F}|$ such that $u_i,u_j\in F_t$ and $\{u_{i+1},\ldots,u_{j-1}\}\cap F_t=\emptyset$, perform the following operation: add a new vertex into $V(G)$, $s_{i,j,t}$, add the edges $\{w_i,s_{i,j,t}\}$ and $\{s_{i,j,t},v_j\}$ into $E(G)$ and let their color-set be $\{t\}$. Moreover, for each $1\leq t\leq |{\cal F}|$, let $u_i$ and $u_j$ be the elements with the largest and smallest index contained in $F_t$, respectively, and perform the following operation: add a new vertex into $V(G)$, $s_{i,j,t}$, add the edges $\{w_i,s_{i,j,t}\}$ and $\{s_{i,j,t},v_j\}$ into $E(G)$, and let their color-set be $\{t\}$. Observe that $|V(G)|=\OO(|U||{\cal F}|)$ and that $\alpha=|{\cal F}|$. Therefore, $\alpha=\OO(|V(G)|)$. It remains to show that the instances $(G,k,{\sf col})$ and $(U,{\cal F},k)$ are equivalent. By construction, each monochromatic cycle in $G$ is of the form $v_{i_1}-w_{i_1}-s_{i_1,i_2,t}-v_{i_2}-w_{i_2}-s_{i_2,i_3,t}-\cdots-v_{i_{|F_t|}}-w_{i_{|F_t|}}-s_{i_{|F_t|},i_1,t}-v_{i_1}$, where $\{u_{i_1},u_{i_2},\ldots,u_{i_{|F_t|}}\}=F_t\in{\cal F}$, and for each set $F_t\in{\cal F}$, $G$ contains exactly one such monochromatic cycle.

\begin{lemma}
  \label{lem:hitting-set-colored-graph}
$(U,{\cal F},k)$ is a \yes\ instance of \textsc{HS} if and only if $(G,k,{\sf col}: E(G) \rightarrow 2^{[\alpha]})$ is a \yes\ instance of \simfes.
\end{lemma}
\begin{proof}
In the forward direction, let $S$ be a solution to $(U,{\cal F},k)$. Define $Q$ as the set of edges $\{\{v_i,w_i\}: u_i\in S\}\subseteq E(G)$. We claim that $Q$ is a solution to
$(G,k,{\sf col})$. Since $|Q|=|S|$, it holds that $|Q|\leq k$. Now, consider a monochromatic cycle $C$ in $G$. Recall that this cycle is of the form
$v_{i_1}-w_{i_1}-s_{i_1,i_2,t}-v_{i_2}-w_{i_2}-s_{i_2,i_3,t}-\cdots-v_{i_{|F_t|}}-w_{i_{|F_t|}}-s_{i_{|F_t|},i_1,t}-v_{i_1}$, where
$\{u_{i_1},u_{i_2},\ldots,u_{i_{|F_t|}}\}=F_t\in{\cal F}$. In particular, observe that $\{\{v_i,w_i\}: u_i\in F_t\}\subseteq E(C)$. Since $S$ is  is a hitting set of ${\cal F}$, it
holds that $S\cap F_t\neq\emptyset$. 
This implies that $Q\cap\{\{v_i,w_i\}: u_i\in F_t\}\neq\emptyset$, and therefore $Q$ is a solution of $(G,k,\textsf{col}: E(G) \rightarrow 2^{[\alpha]})$.

In the reverse direction, let $Q$ be a solution to $(G,k,{\sf col}: E(G) \rightarrow 2^{[\alpha]})$. By the form of each monochromatic cycle in $G$, if $Q$ contains an edge that
includes a vertex of the form $s_{i,j,t}$, such an edge can be replaced by the edge $\{v_i,w_i\}$. Thus, we can assume that $Q$ only contains edges of the form $\{v_i,w_i\}$.
Define $S$ as the set of elements $\{u_i: \{v_i,w_i\}\in Q\}\subseteq U$. We claim that $S$ is a solution to $(U,{\cal F},k)$. Since $|S|\leq |Q|$, it holds that $|S|\leq k$. Now,
recall that for each set $\{u_{i_1},u_{i_2},\ldots,u_{i_{|F_t|}}\}=F_t\in{\cal F}$, $G$ contains a monochromatic cycle of the form
$v_{i_1}-w_{i_1}-s_{i_1,i_2,t}-v_{i_2}-w_{i_2}-s_{i_2,i_3,t}-\cdots-v_{i_{|F_t|}}-w_{i_{|F_t|}}-s_{i_{|F_t|},i_1,t}-v_{i_1}$. Since $Q$ is a solution of $(G,k,{\sf col}: E(G)
\rightarrow 2^{[\alpha]})$, it holds that $Q\cap\{\{v_i,w_i\}: u_i\in F_t\}\neq\emptyset$. This implies that $S\cap F_t\neq\emptyset$.
\end{proof}

\begin{theorem}\label{thm:tightW2}
{\sc Sim-FES} parameterized by $k$, when $\alpha=\OO(|V(G)|)$, is $W[2]$-hard.
\end{theorem}

\subsection{W[1] Hardness of \simfes{} when $\alpha=\mathcal{O}(\log |V(G)|)$}
We modify the reduction given in the proof of Theorem \ref{thm:tightW2} to show that when $\alpha=\OO(\log|V(G)|)$, {\sc Sim-FES} is W[1]-hard with respect to the parameter $k$. This result implies that the dependency on $\alpha$ of our $\OO((2^{\mathcal{O}(\alpha)})^kn^{\OO(1)})$-time algorithm for {\sc Sim-FES} is optimal in the sense that it is unlikely that there exists an $\OO((2^{o(\alpha)})^kn^{\OO(1)})$-time algorithm for this problem.

We give a reduction from a variant of {\sc HS}, called {Partitioned Hitting Set ({\sc PHS})}, to {\sc Sim-FES} where $\alpha=\OO(\log|V(G)|)$. The input of {\sc PHS} consists of a universe $U$, a collection ${\cal F}=\{F_1,F_2,\ldots,F_{|{\cal F}|}\}$, where each $F_i$ is a family of {\em disjoint} subsets of $U$, and a parameter $k$. The goal is to decide the existence of a subset $S\subseteq U$ of size at most $k$ such that for all $f\in(\bigcup_{F \in \cal F} F)$, $S\cap f\neq\emptyset$. It is known that the special case of {\sc PHS} where $|{\cal F}|=\OO(\log(|U||(\bigcup{\cal F})|))$ is W[1]-hard when parameterized by $k$ (see, e.g.,~\cite{Sim-FVS}). Thus, to prove the theorem, it is sufficient to construct (in polynomial time) an instance of the form $(G,k,{\sf col}: E(G) \rightarrow 2^{[\alpha]})$ of {\sc Sim-FES} that is equivalent to $(U,{\cal F},k)$, where $\alpha=\OO(\log|V(G)|)$. 
The construction of the graph $G$ is exactly similar to the one in Theorem~\ref{thm:tightW2}. But instead of 
creating a unique monochromatic cycle with a color $i$ for each $f_i\in \bigcup{\cal F}$, for each $F_i\in {\cal F}$ 
we create $\vert F_i \vert$ vertex disjoint cycles of same color $i$. Since for each $F\in {\cal F}$ the sets 
in $F$ are pairwise disjoint, guarantees the correctness. Formal description of the reduction is given below. 

Without loss of generality we assume that each set in $\bigcup_{F \in \cal F} F$ contains at least two elements from $U$. The instance $(G,k,{\sf col}: E(G) \rightarrow 2^{[\alpha]})$ is constructed as follows. Initially, $V(G)=E(G)=\emptyset$. For each element $u_i\in U$, insert two new vertices $v_i$ and $w_i$ into $V(G)$, and add the edge $\{v_i,w_i\}$ into $E(G)$ with its color-set being $\{j: F_j\in {\cal F}, u_i\in (\bigcup F_j)\}$. Now, for all $1\leq i<j\leq |U|$ and for all $1\leq t\leq |\mathcal{F}|$ such that there exists $f\in F_t$ satisfying $u_i,u_j\in f$ and $\{u_{i+1},\ldots,u_{j-1}\}\cap f=\emptyset$, perform the following operation: add a new vertex $s_{i,j,t}$ into $V(G)$, add the edges $\{w_i,s_{i,j,t}\}$ and $\{s_{i,j,t},v_j\}$ into $E(G)$ with both of its color-set being $\{t\}$. Moreover, for each $1\leq t\leq |{\cal F}|$ and $f\in F_t$, let $u_i$ and $u_j$ be the elements with the largest and smallest index contained in $f$, respectively, we perform the following operation: add a new vertex into $V(G)$, $s_{i,j,t}$, add the edges $\{w_i,s_{i,j,t}\}$ and $\{s_{i,j,t},v_j\}$ into $E(G)$, and let their color-set be $\{t\}$. Observe that $|V(G)|=\OO(|U||(\bigcup{\cal F})|)$ and that $\alpha=|{\cal F}|$. Since $|{\cal F}|=\OO(\log(|U||(\bigcup{\cal F})|))$, we have that $\alpha=\OO(\log|V(G)|)$. Since the sets in each family $F_i$ are disjoint, the construction implies that each monochromatic cycle in $G$ is of the form $v_{i_1}-w_{i_1}-s_{i_1,i_2,t}-v_{i_2}-w_{i_2}-s_{i_2,i_3,t}-\cdots-v_{i_{|f|}}-w_{i_{|f|}}-s_{i_{|f|},i_1,t}-v_{i_1}$, where $\{u_{i_1},u_{i_2},\ldots,u_{i_{|F_t|}}\}=f$ for a set $f\in F_t\in{\cal F}$, and for each set $f\in F_t\in{\cal F}$, $G$ contains a monochromatic cycle of this form. By using the arguments similar to one in the proof of Lemma~\ref{lem:hitting-set-colored-graph}, we get that the instances $(G,k,{\sf col}: E(G) \rightarrow 2^{[\alpha]})$ and $(U,{\cal F},k)$ are equivalent. Hence we get the  following theorem.

\begin{theorem}\label{thm:tightW1}
{\sc Sim-FES} parameterized by $k$, when $\alpha=\OO(\log|V(G)|)$ is $W[1]$-hard.
\end{theorem}

\section{Kernel for \simfesfull{}}\label{sec:kernel}
In this section we give a kernel for {\sc Sim-FES} with $\OO((k\alpha)^{\mathcal{O}(\alpha)})$ vertices. We start by applying preprocessing rules similar in spirit to the ones used to obtain a kernel for {\sc Feedback Vertex Set}, but it requires subtle differences due to the fact that we handle a problem where edges rather than vertices are deleted, as well as the fact that the edges are colored (in particular, each edge in {\sc Sim-FES} has a color-set, while each vertex in {\sc Sim-FVS} is uncolored). We obtain an approximate solution by computing a spanning tree per color. We rely on the approximate solution to bound the number of vertices whose degree in certain subgraphs of $G$ is not equal to $2$. Then, the number of the remaining vertices is bounded by adapting the ``interception''-based approach of Agrawal et al.~\cite{Sim-FVS} to a form relevant to {\sc Sim-FES}.

Let $(G,k,{\sf col}: E(G) \rightarrow 2^{[\alpha]})$ be an instance of {\sc Sim-FES}. For each color $i \in [\alpha]$ recall $G_i$ is the graph consisting of the vertex-set $V(G)$ and the edge-set $E(G_i)$ includes every edge in $E(G)$ whose color-set contains the color $i$. It is easy to verify that the following rules are correct when applied exhaustively in the order in which they are listed. We note that the resulting instance can contain multiple edges.

\begin{itemize}
\item {\bf Reduction Rule 1:} If $k < 0$, return that $(G,k,{\sf col}: E(G) \rightarrow 2^{[\alpha]})$ is a \No\ instance.
\item {\bf Reduction Rule 2:} If for all $i\in [\alpha]$, $G_i$ is acyclic, return that $(G,k,{\sf col}: E(G) \rightarrow 2^{[\alpha]})$ is a \Yes\ instance.
\item {\bf Reduction Rule 3:} If there is a self-loop at a vertex $v \in V(G)$, then remove $v$ from $G$ and decrement $k$ by 1. 
\item {\bf Reduction Rule 4:} If there exists an isolated vertex in $G$, then remove it. 
\item {\bf Reduction Rule 5:} If there exists $i \in [\alpha]$ and an edge whose color-set contains $i$ but it does not participate in any cycle in $G_i$, remove $i$ from its color-set. If the color-set becomes empty, remove the edge.
\item {\bf Reduction Rule 6:} If there exists $i \in [\alpha]$ and a vertex $v$ of degree exactly two in $G$, remove $v$ and connect its two neighbors by an edge whose color-set is the same as the color-set of the two edges incident to $v$ 
(we prove in Lemma~\ref{lem:simfes-deg-2-rule} that the color set of two edges are same).
\end{itemize}

\begin{lemma} 
Reduction rule 6 is safe.
\label{lem:simfes-deg-2-rule}
\end{lemma}

\begin{proof}
 Let $G$ be a graph with coloring function \colfn. Let $v$ be a vertex in $V(G)$ such that $v$ has total degree $2$ in $G$. We have applied Reduction Rule 1 to 5 exhaustively (in
that order). Therefore, when Rule 6 is applied, the edges incident to $v$ have the same color-set say $i$, since otherwise Rule 5 would be applicable. Let $u,w$ be the neighbors of
$v$ in $G_i$, where $i \in [\alpha]$. Consider the graph $G'$ with vertex set as $V(G)\setminus \{v\}$ and edge set as $E(G')=(E(G) \setminus \{\{v,u\}, \{v,w\}\}) \cup 
\{\{u,w\}\}$ and coloring
function ${\sf col}'$ such that ${\sf col}'(\{u,w\})= {\sf col}(\{u,w\}) \cup \{i\}$ and for all other edges $e\in E(G') \setminus \{\{u,w\}\}$, ${\sf col}'(e)= {\sf col}(e)$. We
show that $(G,k,$ \colfn$)$ is a \Yes\ instance of \simfes\ if 
and only if $(G',k, {\sf col}')$ is a \Yes\ instance of \simfes.

In the forward direction, let $S$ be a solution to \simfes\ in $G$ of size at most $k$. Suppose $S$ is not a solution in $G'$. Then, there is a cycle $C$ in $G'_t$, for some $t \in
[\alpha]$. Note that $C$ cannot be a cycle in $G'_j$ as $G'_j=G_j$, for $j \in [\alpha] \setminus \{i\}$. Therefore $C$ must be a cycle in $G'_i$. All the cycles $C'$ not
containing the edge $\{u,w\}$ are also cycles in $G_i$ and therefore $S$ must contain some edge from $C'$. It follows that $C$ must contain the edge $\{u,w\}$. Note that the edges
$(E(C) \setminus \{\{u,w\}\}) \cup \{\{v,u\},\{w,v\}\}$ form a cycle in $G_i$. Therefore $S$ must contain an edge from $E(C) \cup \{\{v,u\},\{w,v\}\}$. We consider the following
cases:
\begin{itemize}
 \item Case 1: $\{v,u\},\{w,v\} \notin S$. In this case $S$ must contains an edge from $E(C)\setminus \{\{u,w\}\}$. Hence, $S$ is a solution in $G'$.
 \item Case 2: At least one of $\{v,u\},\{w,v\}$ belongs to $S$, say $\{v,u\} \in S$. Let $S'=(S \setminus \{\{v,u\}\}) \cup \{\{u,w\}\}$. Observe that $S'$ intersects all cycles
in  $G'_i$. Therefore $S'$ is a solution in $G'$ of size at most $k$.
\end{itemize}

In the reverse direction, consider a solution $S$ to \simfes\ in $G'$. If $S$ is a solution in $G$ we have a proof of the claim. Therefore, we assume that $S$ is not a solution in
$G$. $S$ intersects all cycles in $G_j$, since $G_j=G'_j$, for all $j \in [\alpha] \setminus \{i\}$. All cycles in $G_i$ not containing $v$ are also cycles in $G'_i$ and therefore
$S$ intersects all such cycles. 

A cycle in $G_i$ containing $v$ must contain $u$ and $w$ ($v$ is a degree-two vertex in $G_i$). We assume that there is a cycle $C$ containing $v$ in $G$ such that $S$ does not
intersect $C$ (otherwise $S$ is a solution in $G$). Note that in $G'_i$ we added an edge $\{u,w\}$ and we keep multi-edges. Corresponding to each copy of $\{u,w\}$ we have a cycle
in $G'$ with edges $(E(C) \setminus \{\{v,u\}, \{v,w\}\}) \cup \{u,w\}$. Therefore, $S$ must contain all copies of $\{u,w\}$. We create a solution $S'$ by replacing a copy (with
same color as $\{v,u\}$) of $\{u,w\} \in S$ by $\{v,u\}$. We claim that $S'$ is a solution in $G$. $S'$ intersects all the cycles in $G$ containing $v$. Observe that all cycles in
$G$ not containing $v$ are also cycles in $G'$ and they do not contain the deleted copy of the edge $\{u,w\}$ from $S$. Therefore, they are intersected by $S'$. Therefore, $S'$ is
a solution in $G$ of size at most $k$.
\end{proof}

We apply Reduction Rule 1 to 6 exhaustively (in that order). The safeness of Reduction Rules 1 to 5 are easy to see. 
Lemma~\ref{lem:simfes-deg-2-rule} proves the safeness Reduction Rule 6. After this, we follow the approach similar to that in~\cite{Sim-FVS} to bound the size of the instance. This
gives the following theorem.

\begin{theorem}\label{thm:kernel}
{\sc Sim-FES} admits a kernel with $(k\alpha)^{\mathcal{O}(\alpha)}$ vertices.
\end{theorem}
\begin{proof}
Let $(G,k,{\sf col}: E(G) \rightarrow 2^{[\alpha]})$ be an instance of \simfes{} where none of the Reduction rules are applicable. For each graph $G_i$, we compute a spanning
forest, $F_i$, maximizing $|E(F_i)|$. Let $X_i = E(G_i) \setminus E(F_i)$. If $|X_i| > k$, the instance is a no-instance. Thus, we can assume that for each $i \in [\alpha]$, $X_i$
contains at most $k$ edges. Let $X=\bigcup_{i=1}^\alpha X_i$ denote the union of the sets $X_i$. Clearly, $|X| \leq k\alpha$. Let $U$ denote the subset of $V(G)$ that contains the
vertices incident to at least one edge in $X$. Since Reduction Rule 5 is not applicable, therefore $|U|\leq 2k\alpha$. Thus, the number of leaves in each $G_i - X$ is bounded by
$2k\alpha$. Accordingly, the number of vertices in each $G_i- X$ whose degree is at least 3 is bounded by $2k\alpha$. It remains to bound the number of vertices that are not
incident to any edge in $X$ and whose degree in each $G_i$ is $0$ or $2$ (their degree in $G$ is at least 3). Let $T$ be the set of vertices in $G$ which is either a leaf or a
degree $3$ vertex in some $G_i$, for $i \in [\alpha]$. Denote the set of vertices which are not in $T$, not incident to any edge in $X$ and whose degree in $G_i$ is $2$ by $D_i$.
Let ${\cal P}_i$ denote the set of paths in $G_i$, for $i \in [\alpha]$, whose internal vertices belong to $D_i$ and whose first and last vertices do not belong to $D_i$. Moreover,
let $D=\bigcup_{i=1}^{\alpha} D_i$ and ${\cal P} = \bigcup_{i=1}^{\alpha} {\cal P}_i$. Observe that for $i\in [\alpha]$, $|{\cal P}_{i}| \leq 4k \alpha$ and $|{\cal P}|\leq
4k\alpha^{2}$.

To obtain the desired kernel, it remains to show that $|D|=\OO((k\alpha)^{\mathcal{O}(\alpha)})$. For each edge $e\in E(G)$, let ${\cal P}[e]$ be the set of paths in ${\cal P}$ to
which $e$ belongs. Each edge belongs to at most one path in each ${\cal P}_i$, for any $i\in [\alpha]$. For each $v\in D$, by $E(v)$ we denote the set of edges incident to $v$ in
$G$. Observe that each vertex in $D$ is incident to at most $2\alpha$ edges. 
 For each vertex $v\in D$, there are at most $(4k\alpha+1)^{\alpha}$ options of choosing to which paths in ${\cal P}$ the vertex $v$ belongs. Note that here the extra additive one
is to include the case when a vertex does not belong to any path in a color class. 
the edges incident to $v$. Thus, if $|D| \geq 3(4k\alpha+1)^{\alpha}$,
  Thus, there exists a constant $c$ such that if $|D|>(k\alpha)^{c\alpha}$, then $D$ contains (at least) three vertices, $r,s$ and $t$, such that for all $q,p\in\{r,s,t\}$, there
is a bijection $f: E(q)\rightarrow E(p)$ such that ${\cal P}[e]={\cal P}[f(e)]$ for all $e\in E(q)$. In particular, if $|D|>(k\alpha)^{c\alpha}$, then $D$ contains two non-adjacent
vertices, $v$ and $u$, such that there is a bijection $f: E(v)\rightarrow E(u)$ satisfying ${\cal P}[e]={\cal P}[f(e)]$ for all $e\in E(v)$. In this case, it is not necessary
insert any edge $e\in E(v)$ into a solution, since it has the same affect as inserting the edge $f(e)$. Thus, we can remove the vertex $v$, and for each two neighbors of $v$, $x$
and $y$, and for each color $i\in[\alpha]$ such that $i\in {\sf col}(\{v,x\})\cap {\sf col}(\{v,y\})$, we insert an edge $\{x,y\}$ whose color-set is $\{i\}$. After an exhaustive
application of this operation (as well as Reduction Rules 1--6), we obtain the desired bound on $|D|$, which concludes the proof of Theorem
\ref{thm:kernel}.
\end{proof}

\section{\maxsimfesfull}\label{sec:exactalgo-fes}
In this section we design an algorithm for \maxsimfesfull. Let $(G,q,\coln: E(G) \rightarrow 2^{[\alpha]})$ be an input to \maxsimfes. A set $F \subseteq E(G)$ such that  for all $i \in [\alpha]$, $G[F_i]$ is acyclic is called {\em simultaneous forest}. Here, $F_i=\{e\in F~|~i\in \coln(e)\}$, denotes the subset of edges of $F$ which has the integer $i$ in its image when the function $\coln$ is applied to it. 
We will solve \maxsimfes by reducing to an equivalent instance of the \matparity problem and then using the algorithm for the same. 

We start by giving a construction that reduces the \maxsimfes to \matparity. Let $(G,q,\coln: E(G) \rightarrow 2^{[\alpha]})$ be an input to \maxsimfes. Given, $(G,q,\coln: E(G) \rightarrow 2^{[\alpha]})$, for $i\in[\alpha]$, recall that by $G_i$ we denote the graph with the vertex set $V(G_i)=V(G)$ and the edge set $E(G_i)=\{ e^i \mid e\in E(G) \mbox{ and } i\in \coln(e)\}$. For each edge $e\in E(G)$, we will have its  distinct copy in $G_i$ if $i \in \coln(e)$. Thus, for each edge $e\in E(G)$, by  ${\sf Original}(e)$ we denote the set of edges $\{ e^j ~|  j\in \coln(e) \}$. On the other hand for each edge $e \in E(G)$, by  ${\sf Fake}(e)$ we denote the set of edges $\{ e^j ~|  j\in [\alpha]-\coln(e)\}$. Finally, for each edge $e\in E(G)$, by ${\sf Copies}(e)$ we denote the set ${\sf Original}(e) \cup {\sf Fake}(e)$. Let $M_{i}=(E_i,{\cal I}_i)$  denote the graphic matroid on $G_i$. That is, edges of $G_i$ forms the universe $E_i$ and ${\cal I}_i$ contains, $S\subseteq E(G_i)$ such that  $G_i[S]$ forms a forest. By Proposition~\ref{prop:graphicrep} we know that graphic matroids are representable over any field 
and given a graph $G$ one can find the corresponding representation matrix in time polynomial in $|V(G)|$. Let $A_{i}$ denote the linear representation  of $M_i$. That is, $A_i$ is
a matrix over ${\mathbb F}_2$,  where  the set of columns of $A_i$ are denoted by $E(G_i)$.  In particular, $A_i$ has dimension $d \times |E(G_i)|$, where $d=\mrank{M_i}$. A set $X
\subseteq E(G_i)$ is independent (that is $X\in {\cal I}_i$) if and only if the corresponding columns are linearly independent over ${\mathbb F}_2$.  Let ${\sf Fake}(G)$ denote the
set of edges in $\bigcup_{e\in E(G)}{\sf Fake}(e)$. Furthermore, let $\tau=|{\sf Fake}(G)|=\sum_{e\in E(G)} |{\sf Fake}(e)|$. Let $M_{\alpha+1}$ be the uniform matroid over ${\sf
Fake}(G)$ of rank $\tau$. That is, $E_{\alpha+1}={\sf Fake}(G)$  and $M_{\alpha+1}=U_{\tau,\tau}$. Let $I_{\tau}$ denote the  identity matrix of dimension $\tau \times \tau$.
Observe that, $A_{\alpha+1}=I_{\tau}$ denotes the linear representation of $M_{\alpha+1}$ over ${\mathbb F}_2$. Notice that $E_i\cap  E_j =\emptyset$ for all $1\leq i\neq j \leq
\alpha+1$. Let $M$ denote the direct sum of 
$M_1\oplus \cdots \oplus M_{\alpha+1}$ with its representation matrix being $A_M$.  

Now we are ready to define an instance of \matparity.  The ground set is the columns of $A_M$, which is indexed by edges in $\bigcup_{e\in E(G)} {\sf Copies}(e)$. Furthermore, the ground set is partitioned into ${\sf Copies}(e)$, $e\in E(G)$, which are called blocks. The main technical lemma of this section on which the whole algorithm is based is the following. 

\begin{lemma}
\label{lem:maxsimforesttomatparity}
Let $(G,q,\coln: E(G) \rightarrow 2^{[\alpha]})$ be an instance of \maxsimfes. Then $G$ has a simultaneous forest of size $q$ if and only if $(A_M, \biguplus_{e\in E(G)} {\sf
Copies}(e),q)$ is a \yes\ instance of \matparity. Furthermore, given $(G,q,\coln: E(G) \rightarrow 2^{[\alpha]})$ we can obtain an instance $(A_M, \biguplus_{e\in E(G)} {\sf
Copies}(e),q)$ in polynomial time. 
\end{lemma}

\begin{proof}
 We first show the forward direction of the proof. Let $F$ be  a simultaneous forest of size $q$. Then we claim that the columns corresponding to $S=\bigcup_{e\in F}{\sf
Copies}(e)$ form an independent set in $M$ and furthermore, it is the union of $q$ blocks. That is, we need to show that the columns corresponding to $S=\bigcup_{e\in F}{\sf
Copies}(e)$ are linearly independent in $A_M$ over $\mathbb{F}_2$. By the definition of direct sum and its linear representation, it reduces to showing that $F$ is linearly
independent if and only if $F\cap E_i\in {\cal I}_i$ for all $i\leq \alpha+1$. Since $F$ is  a simultaneous forest of size $q$, we have that $G[F_i]$, $F_i=\{e\in F~|~i\in
\coln(e)\}$, is a forest. Hence, this implies that $F^i=\{e^i~|~e\in F_i\}$ forms a forest in $G_i$. This takes care of all the edges in $\cup_{e\in F} {\sf Original}(e)$. Now let
$S^*=S-\cup_{i\in [\alpha]} F^i =\cup_{e\in F} {\sf Fake}(e)=F^{\alpha+1}$. However, $S^*$ is a subset of $U_{\tau, \tau}$ and thus is an independent set since $|S^*|\leq \tau$.
This completes the proof of the forward direction. 

Now we show the reverse direction of the proof. Since, $(A_M, \biguplus_{e\in E(G)} {\sf Copies}(e),q)$ is a yes instance of \matparity, there exists  an independent set, say $S$,
that is the union of $q$ blocks. By construction $S$ corresponds to union of the sets ${\sf Copies}(e)$ for some $q$ edges in $G$. Let these edges be $F=\{e_1,\ldots,e_q\}$. We 
claim that $F$ is a simultaneous forest of size $q$. Towards this, we need to show that $G[F_i]$, where $F_i=\{e\in F~|~i\in \coln(e)\}$, is a forest. This happens if and only if
$F^i=\{e^i~|~e\in F_i\}$ forms a forest in $G_i$. However, we know that the columns corresponding to $F_i$ are linearly independent in $M_i$ and in particular in $A_i$ -- the
linear representation of graphic matroid of $G_i$. This shows that $F_i$ forms a forest in $G_i$ and hence $G[F_i]$ is a forest. This completes the equivalence proof. 

Finally, it easily follows from the discussion preceding the lemma that given $(G,q,\coln: E(G) \rightarrow 2^{[\alpha]})$ we can obtain an instance $(A_M, \biguplus_{e\in E(G)}
{\sf Copies}(e),q)$  in time polynomial  in $|V(G)|$. This completes the proof of the lemma.
\end{proof}

We will use the polynomial time reduction provided in Lemma~\ref{lem:maxsimforesttomatparity} to get the desired \FPT\ algorithm for \maxsimfes. Towards this will use the following \FPT{} result regarding \matparity for our \FPT{} as well as for an exact exponential time algorithm.  

Given an instance $(G,q,\coln: E(G) \rightarrow 2^{[\alpha]})$ of \maxsimfes\ we first apply Lemma~\ref{lem:maxsimforesttomatparity} and  obtain an instance $(A_M, \biguplus_{e\in E(G)} {\sf Copies}(e),q)$ of \matparity\ and then apply Proposition~\ref{alg:matparity} to obtain the following result. 

\begin{theorem}
\label{thm:maxsimfes-fpt}
\maxsimfes can be solved in time $\cO(2^{\omega q\alpha} |V(G)|^{\cO(1)})$. 
\end{theorem}

Let $(G,q,\coln: E(G) \rightarrow 2^{[\alpha]})$ be an instance of \maxsimfes. Observe that $q$ is upper bounded by $\alpha (|V(G)|-1)$. Thus, as a corollary to Theorem~\ref{thm:maxsimfes-fpt} we get an exact algorithm for finding the largest sized simultaneous acyclic 
subgraph, running in time $\cO(2^{\omega n \alpha^2} |V(G)|^{\cO(1)})$.

\bibliographystyle{plain}
\bibliography{main_simfes}

\begin{thebibliography}{10}

\bibitem{Sim-FVS}
{\sc A~Agrawal, D~Lokshtanov, A~E Mouawad, and S~Saurabh}, {\em Simultaneous
  feedback vertex set: {A} parameterized perspective}, in 33rd Symposium on
  Theoretical Aspects of Computer Science, {STACS}, 2016, pp.~7:1--7:15.

\bibitem{alon-homorphism}
{\sc NOGA Alon and Timothy~H Marshall}, {\em Homomorphisms of edge-colored
  graphs and coxeter groups}, Journal of Algebraic Combinatorics, 8 (1998),
  pp.~5--13.

\bibitem{2-approx-fvs-bafna}
{\sc Vineet Bafna, Piotr Berman, and Toshihiro Fujito}, {\em A 2-approximation
  algorithm for the undirected feedback vertex set problem}, SIAM Journal on
  Discrete Mathematics, 12 (1999), pp.~289--297.

\bibitem{balogh-partitioning}
{\sc J{\'o}zsef Balogh, J{\'a}nos Bar{\'a}t, D{\'a}niel Gerbner, Andr{\'a}s
  Gy{\'a}rf{\'a}s, and G{\'a}bor~N S{\'a}rk{\"o}zy}, {\em Partitioning
  2-edge-colored graphs by monochromatic paths and cycles}, Combinatorica, 34
  (2014), pp.~507--526.

\bibitem{jensen-alternating}
{\sc J{\o}rgen Bang-Jensen and Gregory Gutin}, {\em Alternating cycles and
  paths in edge-coloured multigraphs: a survey}, Discrete Mathematics, 165
  (1997), pp.~39--60.

\bibitem{caiye2014}
{\sc Leizhen Cai and Junjie Ye}, {\em Dual connectedness of edge-bicolored
  graphs and beyond}, in Mathematical Foundations of Computer Science, {MFCS},
  vol.~8635, 2014, pp.~141--152.

\bibitem{chou-paths}
{\sc WS~Chou, Yannis Manoussakis, Olga Megalakaki, M~Spyratos, and Zs~Tuza},
  {\em Paths through fixed vertices in edge-colored graphs}, Math{\'e}matiques
  et sciences humaines, 127 (1994), pp.~49--58.

\bibitem{saket-book}
{\sc M~Cygan, F~V Fomin, L~Kowalik, D~Lokshtanov, D~Marx, M~Pilipczuk,
  M~Pilipczuk, and S~Saurabh}, {\em Parameterized algorithms}, Springer, 2015.

\bibitem{Cygan:2011:SCP:2082752.2082943}
{\sc Marek Cygan, Jesper Nederlof, Marcin Pilipczuk, Micha{\l} Pilipczuk,
  Joham~MM van Rooij, and Jakub~Onufry Wojtaszczyk}, {\em Solving connectivity
  problems parameterized by treewidth in single exponential time}, in
  Foundations of Computer Science (FOCS), IEEE 52nd Annual Symposium, 2011,
  pp.~150--159.

\bibitem{diestel-book}
{\sc Reinhard Diestel}, {\em Graph Theory, 4th Edition}, vol.~173 of Graduate
  texts in mathematics, Springer, 2012.

\bibitem{exactmonotoneFomin}
{\sc Fedor~V. Fomin, Serge Gaspers, Daniel Lokshtanov, and Saket Saurabh}, {\em
  Exact algorithms via monotone local search}, in Proceedings of the 48th
  Annual ACM SIGACT Symposium on Theory of Computing, STOC, 2016, pp.~764--775.

\bibitem{GJ79}
{\sc M~R Garey and D~S Johnson}, {\em Computers and intractability: A guide to
  the theory of NP-completeness}, W. H. Freeman \& Co., 1979.

\bibitem{Impagliazzo:2001:PSE:569473.569474}
{\sc Russell Impagliazzo, Ramamohan Paturi, and Francis Zane}, {\em Which
  problems have strongly exponential complexity?}, Journal of Computer and
  System Sciences, 63 (2001), pp.~512--530.

\bibitem{kano-monochromatic}
{\sc Mikio Kano and Xueliang Li}, {\em Monochromatic and heterochromatic
  subgraphs in edge-colored graphs-a survey}, Graphs and Combinatorics, 24
  (2008), pp.~237--263.

\bibitem{Kociumaka2014556}
{\sc Tomasz Kociumaka and Marcin Pilipczuk}, {\em Faster deterministic feedback
  vertex set}, Information Processing Letters, 114 (2014), pp.~556 -- 560.

\bibitem{komusiewicz2015tight}
{\sc Christian Komusiewicz}, {\em Tight running time lower bounds for vertex
  deletion problems}, arXiv preprint arXiv:1511.05449,  (2015).

\bibitem{lokshtanov2015deterministic}
{\sc Daniel Lokshtanov, Pranabendu Misra, Fahad Panolan, and Saket Saurabh},
  {\em Deterministic truncation of linear matroids}, in Automata, Languages,
  and Programming, Springer, 2015, pp.~922--934.

\bibitem{lovasz1980matroid}
{\sc L{\'a}szl{\'o} Lov{\'a}sz}, {\em Matroid matching and some applications},
  Journal of Combinatorial Theory, Series B, 28 (1980), pp.~208--236.

\bibitem{manoussakis-alternating}
{\sc Yannis Manoussakis}, {\em Alternating paths in edge-colored complete
  graphs}, Discrete Applied Mathematics, 56 (1995), pp.~297--309.

\bibitem{marxmatroid06}
{\sc D\'{a}niel Marx}, {\em A parameterized view on matroid optimization
  problems}, Theoretical Computer Science, 410 (2009), pp.~4471--4479.

\bibitem{mohar2001face}
{\sc Bojan Mohar}, {\em Face covers and the genus problem for apex graphs},
  Journal of Combinatorial Theory, Series B, 82 (2001), pp.~102--117.

\bibitem{oxleymatroid}
{\sc James~G Oxley}, {\em Matroid theory}, vol.~3, Oxford University Press,
  USA, 2006.

\bibitem{Thomasse:2010:KKF:1721837.1721848}
{\sc St{\'e}phan Thomass{\'e}}, {\em A $4k^2$ kernel for feedback vertex set},
  ACM Transactions on Algorithms, {(TALG)}, 6 (2010), pp.~32:1--32:8.

\bibitem{Williams12}
{\sc Virginia~Vassilevska Williams}, {\em Multiplying matrices faster than
  coppersmith-winograd}, in Proceedings of the 44th annual ACM symposium on
  Theory of computing, 2012, pp.~887--898.

\end{thebibliography}

\end{document}